\documentclass[nomath,reqno,a4paper,12pt]{amsart}

\makeatletter \let\th@plain\relax \makeatother

\usepackage{mystyle}

\usepackage[utf8]{inputenc}
\usepackage[T1]{fontenc}
\usepackage{lmodern}
\usepackage[british]{babel}
\usepackage{microtype}
\usepackage[headings]{fullpage}
\usepackage{parskip}

\usepackage[pdfencoding=auto,pdfusetitle]{hyperref}
\hypersetup{colorlinks=true, linkcolor=blue!30!black, citecolor=green!30!black, urlcolor=red!45!black}

\usepackage{amsfonts,amssymb,bbm,stmaryrd}
\usepackage[scr=boondoxo,scrscaled=1.05]{mathalfa}
\usepackage{amsmath}
\usepackage[thmmarks, amsmath, amsthm]{ntheorem}
\theorempreskip{\parsep}
\AtBeginEnvironment{proof}{\vspace{\parsep}}

\usepackage{mathtools}
\usepackage[noabbrev,capitalise]{cleveref}
\usepackage{centernot,float,subcaption,url,array,booktabs,colortbl}

\usepackage[shortlabels]{enumitem}
\setlist[enumerate,1]{label={\roman*)}}

\usepackage{tikz}
\usetikzlibrary{matrix,positioning}
\usetikzlibrary{calc}
\usetikzlibrary{arrows}
\tikzset{> =stealth}

\tikzset{vertex/.style={circle,fill=black,minimum size=3mm,inner sep=0mm}}

\newcommand{\addQEDstyle}[2]{\AtBeginEnvironment{#1}{\pushQED{\qed}\renewcommand{\qedsymbol}{#2}}\AtEndEnvironment{#1}{\popQED}}

\theoremstyle{plain}
\newtheorem{theorem}{Theorem}[section]

\newtheorem{proposition}[theorem]{Proposition}

\theoremstyle{definition}

\theoremstyle{remark}
\newtheorem{remark}[theorem]{Remark}
\addQEDstyle{remark}{$\triangle$}

\addQEDstyle{example}{$\triangle$}

\renewcommand{\epsilon}{\varepsilon}
\renewcommand{\phi}{\varphi}

\mathchardef\mhyphen="2D

\renewcommand{\O}{\mathcal{O}}
\renewcommand{\P}{\mathsf{P}}
\newcommand{\NP}{\mathsf{NP}}
\newcommand{\coNP}{\mathsf{coNP}}
\newcommand{\UP}{\mathsf{UP}}
\newcommand{\coUP}{\mathsf{coUP}}
\newcommand{\NL}{\mathsf{NL}}
\newcommand{\PL}{\mathsf{PL}}
\newcommand{\CC}{\mathsf{CC}}
\newcommand{\AC}{\mathsf{AC}}
\newcommand{\parityL}{{\oplus}\mathsf{L}}

\title{A simple lower bound for ARRIVAL}

\author[G. Manuell]{Graham Manuell}
\address{Centre for Mathematics, University of Coimbra, Coimbra, Portugal}
\email{graham@manuell.me}
\thanks{The author acknowledges financial support from the Centre for Mathematics of the University of Coimbra (UIDB/00324/2020, funded by the Portuguese Government through FCT/MCTES)}

\date{August 2021}


\begin{document}

\begin{abstract}
 The \textsc{Arrival} problem introduced by Dohrau, Gärtner, Kohler, Matoušek and Welzl concerns a train moving on a directed graph proceeding along outward edges according to the position of `switches' at each vertex, which in turn are toggled whenever the train passes through them. The problem asks whether the train every reaches a designated destination vertex.
 It is known that $\textsc{Arrival}$ is contained in $\UP \cap \coUP$, while the previously best published lower bound is that it is $\NL$-hard.
 In this note we provide a simple reduction to the $\textsc{Digicomp}_\mathsf{EXP}$ problem considered by Aaronson. It follows in particular that \textsc{Arrival} is both $\CC$-hard and $\PL$-hard.
\end{abstract}

\maketitle
\thispagestyle{empty}

\section{Introduction}

A \emph{switch graph} consists of a set of vertices $V$ and two endomorphisms $s_0, s_1\colon V \to V$ and can be thought of as a directed graph with directed edges from each $v$ to $s_0(v)$ and from $v$ to $s_1(v)$.
Given a switch graph $G = (V, s_0, s_1)$, we imagine a train starting a some vertex $s$ and traversing the graph is the following way. Each vertex of the graph contains a switch initialised to state $0$. At each time step, if the train is at a vertex $v$ and the switch at $v$ is in state $i \in \{0,1\}$, the train moves from a vertex $v$ to $s_i(v)$ and the state of the switch at $v$ is toggled. 
The \textsc{Arrival} problem \cite{dohrau2017arrival} asks whether the train every reaches a specified destination vertex $t$.

It is shown in \cite{dohrau2017arrival} that \textsc{Arrival} is in $\NP \cap \coNP$ and this is improved to $\UP \cap \coUP$ in \cite{gartner2018CPS}. Recently, an algorithm has been given that solves it in $2^{\O(\sqrt{n} \log{n})}$ time \cite{gartner2021subexponential}.
On the other hand, the only published lower bound is given in \cite{fearnley2021reachability} (for a formally slightly more general game) where $\NL$-hardness is proved.

However, in \cite{aaronson2014digicomp} Aaronson studies a highly related problem. The Digi-Comp II was a mechanical toy computer where small balls rolls down an incline and are deflected by toggles that divert their path, which conversely causes the toggles to be kicked back into a different setting. The effect is that the balls behave exactly the train in \textsc{Arrival} with the graph $G$ restricted to be acyclic.

Explicitly, let us define a switch graph to be \emph{acyclic} if there are no cycles in the corresponding directed graph aside from self-loops.
Then the \textsc{Digicomp} problem asks for an acyclic switch graph $G = (V, s_0, s_1)$, a starting vertex $s$, a destination vertex $t$ and a number of balls $T$ encoded in unary, whether after $T$ balls are released sequentially from $s$ any ever reach $t$. (Technically, to match the original definition there are should some further restrictions, such as $s_0(t) = s_1(t) = t$, but our version is easily seen to be equivalent.)
Aaronson shows that this problem is in fact $\CC$-complete.

Recall that $\CC$ is the class of problems which are log-space reducible to the evaluation of a circuit built about of comparator gates (which send $(x,y)$ to $(x \wedge y, x \vee y)$) without implicit fan-out (see \cite{mayr1992fanout,cook2018comparator}). It is known that $\CC$ contains $\NL$.

We note it is not completely immediate that \textsc{Arrival} is $\CC$-hard: while the \textsc{Digicomp} problem is almost \textsc{Arrival} restricted to acyclic graphs, this is not quite true, since unlike \textsc{Arrival}, $\textsc{Digicomp}$ involves `multiple trains'.
Nonetheless, it is still true that \textsc{Digicomp} can be reduced \textsc{Arrival} as we show below.

In fact, we can say say more. The problem $\textsc{Digicomp}_\mathsf{EXP}$ is defined in \cite{aaronson2014digicomp} in the same way as \textsc{Digicomp} except the number of balls $T$ is encoded in \emph{binary} --- that is, the number of balls used can be exponential in the input size. This problem is still in $\P$, but it might be strictly harder than \textsc{Digicomp}. We will show that $\textsc{Digicomp}_\mathsf{EXP}$ can be reduced \textsc{Arrival} in logarithmic space.
Consequently, \textsc{Arrival} is both $\CC$-hard and $\PL$-hard.

I do not claim any particular novelty here; every nontrivial piece of the argument used to derive this result is already known to others. However, they do not been appear to have been put together explicitly before and all the literature on \textsc{Arrival} only mentions the $\NL$-hardness result. The aim of this paper is to make the stronger lower bounds more widely known.

\section{Results}

As noted in the original paper \cite{dohrau2017arrival}, it is not difficult to construct switch graphs which act as binary counters.
A 4-bit counter is given in the following diagram, where solid lines represent the $s_0$ edges and dashed lines represent the $s_1$ edges.

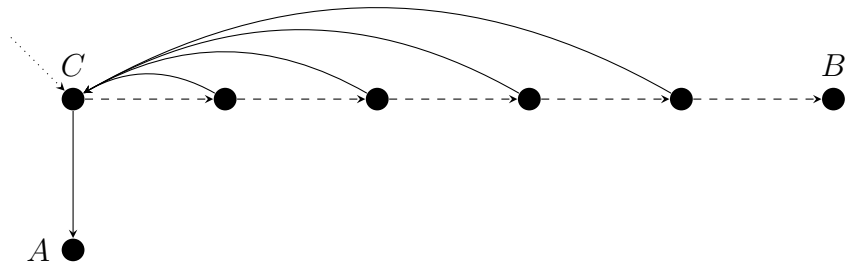
\begin{center}
\begin{tikzpicture}[node distance=2cm,auto]
  \coordinate[vertex,label=above:$C$] (C0);
  \foreach \i in {1,2,3,4} {
    \coordinate[vertex,right of=C\the\numexpr\i-1\relax] (C\i);
    \draw[dashed,->] (C\the\numexpr\i-1\relax) to (C\i);
    \draw[bend right,->] (C\i) to (C0);
  }
  \coordinate[vertex,right of=C4,label=above:$B$] (B);
  \draw[dashed,->] (C4) to (B);
  \coordinate[vertex,below of=C0,label=left:$A$] (A);
  \draw[->] (C0) to (A);
  \coordinate[above left=1cm of C0] (D);
  \draw[dotted,->] (D) to (C0);
\end{tikzpicture}
\end{center}

This can be used to pass down the path at $A$ 16 times and then the path at $B$ once (before possibly repeating the whole sequence over again). Here we have omitted the full path starting at $A$, which would of course need to eventually loop back to the starting vertex of the counter at $C$.

In general, given some counter similar to that above that counts to $T$, we can make $B$ connect back to $C$ via a solid line and to a new final state $B'$ via a dotted line.
This will cause the counter to run twice before reaching $B'$ and hence we obtain a counter that counts to $2T$.

On the other hand, if we connect $B$ to $A$ instead of looping directly back to $C$, we will end up visiting $A$ one additional time before starting the counter again.
Thus, the new counter will count up to $2T+1$.

Now by repeatedly applying these two techniques, we can build a counter to count to any number based on its binary expansion.
For instance, the following counter visits $A$ $22 = 10110_2$ times before reaching $B$.

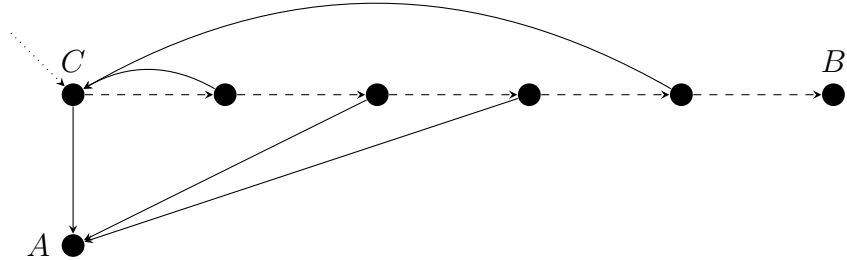
\begin{center}
\begin{tikzpicture}[node distance=2cm,auto]
  \coordinate[vertex,label=above:$C$] (C0);
  \foreach \i in {1,2,3,4} {
    \coordinate[vertex,right of=C\the\numexpr\i-1\relax] (C\i);
    \draw[dashed,->] (C\the\numexpr\i-1\relax) to (C\i);
  }
  
  \coordinate[vertex,right of=C4,label=above:$B$] (B);
  \draw[dashed,->] (C4) to (B);
  \coordinate[vertex,below of=C0,label=left:$A$] (A);
  \draw[->] (C0) to (A);
  \coordinate[above left=1cm of C0] (D);
  \draw[dotted,->] (D) to (C0);
  
  \draw[bend right,->] (C1) to (C0);
  \draw[->] (C2) to (A);
  \draw[->] (C3) to (A);
  \draw[bend right,->] (C4) to (C0);
\end{tikzpicture}
\end{center}

These counters only require $\lfloor\log T\rfloor + 1$ nodes to count up to $T$.
We now arrive at the following result, where we use a counter to run through an acyclic graph multiple times and hence simulate multiple trains.
\begin{proposition}
 ${\normalfont\textsc{Digicomp}}_\mathsf{EXP}$ is a $\AC^0$-reducible to {\normalfont\textsc{Arrival}}.
\end{proposition}
\begin{proof}
 Consider the $\textsc{Digicomp}_\mathsf{EXP}$ problem for an acyclic switch graph $G$, starting vertex $s$, destination vertex $t$ and a number of balls $T$.
 We construct a counter up to $T$ as described above where the output $A$ connects to the graph $G$ at the vertex $s$ and the output at $B$ is a new vertex $F$ with $s_0(F) = s_1(F) = F$.
 The `leaf' vertices of $G$ (i.e.\ those vertices $v$ such that $s_0(v) = s_1(v) = v$) have their outputs modified so that they connect back to the counter at $C$.
 The new starting vertex is $C$ and the new destination vertex is still $t$.
\end{proof}

\begin{remark}
 We have seen it is easy to repeat an operation exponentially many times in \textsc{Arrival} and then do something else. This might make us hope that we can simulate the succinct 0-player reachability switching game of \cite{fearnley2021reachability}, which was shown there to be $\P$-hard. However, no simple modification of our approach is able to achieve this.
 
 The problem is that to simulate Boolean circuits as in \cite{fearnley2021reachability} it is necessary to put the entire program in a big loop as we have done here and have the train pass through the circuit multiple times in order to feed the inputs into the circuit. However, then if we ever use a counter to make the train go down one path $2^n$ times before going along another path, by the time we reach the counter again it will have been reset to its initial position and the train will go down the first path yet again instead of starting to repeat the second path.
 
 We could try to avoid this by making sure the second path never loops back to this counter again, but this is not compatible with the global loop.
 We also note that it \emph{is}, of course, possible to alternate between paths $A$ and $B$ so that in the end both have been traversed $2^n$ times, but now it is not clear how to make logic gates work.
 
 Finally, we note in passing that it is possible to simulate the non-succinct 0-player reachability switching game by encoding each node of high out-degree using a binary tree, but this takes exponentially many nodes and so is not helpful for the succinct version.
 \hfill$\triangle$
\end{remark}

Since $\textsc{Digicomp}_\mathsf{EXP}$ is clearly harder than \textsc{Digicomp}, which is $\CC$-complete, it now follows that \textsc{Arrival} is $\CC$-hard.
The precise difficulty of $\textsc{Digicomp}_\mathsf{EXP}$ is not known, but in a comment at \cite{aaronson2014digicomp} Itai Bar-Natan proves that $\textsc{Digicomp}_\mathsf{EXP}$ is $\PL$-hard.
Thus, \textsc{Arrival} is $\PL$-hard too.
We reproduce the argument here for completeness.
\begin{proposition}
 ${\normalfont\textsc{Digicomp}}_\mathsf{EXP}$ is $\PL$-hard.
\end{proposition}
\begin{proof}
 The problem of determining if there are at least $k$ paths from $s$ to $t$ in an directed acyclic graph $G$ is $\PL$-complete.
 By recursively splitting the vertices in the graph we may assume the out-degree of each vertex is at most 2 without loss of generality.
 
 We can now define a new graph $G'$ to have a vertex $(v,i)$ for each $v \in G$ and $i \in \{0,\dots,n-1\}$, where $n$ is the number of vertices in $G$, and an edge from $(u,i)$ to $(v,i+1)$ whenever $u \ne t$ and there is an edge from $u$ to $v$ in $G$ and also when $u = v = t$. Note that the number of paths from $(s,0)$ to $(t,n-1)$ in $G'$ is then precisely equal to the number of paths from $s$ to $t$ in $G$. Moreover, out-degrees of vertices in $G'$ are also at most $2$. This construction can be done in logarithmic space.
 
 Finally, we construct a switch graph from $G'$ by arbitrarily choosing one of the edges out of each vertex to be the $s_0$ edge and the other to be the $s_1$ in the case the out degree is 2. When the degree is less than 2 we make the remaining edges out of $v$ point to a new vertex $F$ which satisfies $s_0(F) = s_1(F) = F$. The vertex $(t,n-1)$ is a special case, which we discuss later. It is not hard to see that if we start $2^n$ balls at $(s,0)$ then the number of balls that reach the vertex $(v,i)$ is equal to $2^{n-i} c(v,i)$ where $c(v,i)$ is the number of paths from $(s,0)$ to $(v,i)$ in $G'$. Thus, $c(t,n-1)$ balls arrive at $(t,n-1)$, which is also equal to the number of paths from $s$ to $t$ in $G$.
 
 Finally, we connect $(t,n-1)$ to a counter that sends the first $k-1$ balls to $F$ and the remaining balls to a new destination vertex $D$. This counter is the \textsc{Digicomp} analogue of the counter for \textsc{Arrival} we described above. The difference is that whenever we would have looped back to the start of the counter we instead just connect the vertex to itself. (When this counter is interpreted in $\textsc{Arrival}$ the counter will connect then back to the start of vertex $(s,0)$.) A ball will now arrive at $D$ precisely if there are at least $k$ paths from $s$ to $t$ in $G$, as required. Finally, note that the construction of the switch graph can be done in logarithmic space and so we are done.
\end{proof}

It is also notable that at the end of the simulation of the switch graph above, the position of the switch at $(t,n-1)$ contains the answer to whether the number of paths from $s$ to $t$ is odd or even, which is a $\parityL$-complete problem. So if we could use this to influence the location of a ball somehow (i.e.\ to cause a later ball to fall into a position no ball would have reached before), this would show that \textsc{Arrival} is $\parityL$-hard.
However, there does not appear to be any obvious way to do this.

\bibliographystyle{abbrv}
\bibliography{references}

\begin{thebibliography}{1}

\bibitem{aaronson2014digicomp}
S.~Aaronson.
\newblock The power of the {Digi-Comp II}, 2014.
\newblock \url{https://www.scottaaronson.com/blog/?p=1902}.

\bibitem{cook2018comparator}
S.~A. Cook, Y.~Filmus, and D.~T.~M. L{\^{e}}.
\newblock The complexity of the comparator circuit value problem.
\newblock arXiv preprint
  \href{https://arxiv.org/abs/1208.2721v3}{arXiv:arXiv:1208.2721}, 2013.

\bibitem{dohrau2017arrival}
J.~Dohrau, B.~G{\"a}rtner, M.~Kohler, J.~Matou{\v{s}}ek, and E.~Welzl.
\newblock {ARRIVAL}: A zero-player graph game in {NP} {$\cap$} {coNP}.
\newblock In M.~Loebl, J.~Ne{\v{s}}et{\v{r}}il, and R.~Thomas, editors, {\em A
  Journey Through Discrete Mathematics: A Tribute to Ji{\v{r}}{\'i}
  Matou{\v{s}}ek}, pages 367--374. Springer, 2017.

\bibitem{fearnley2021reachability}
J.~Fearnley, M.~Gairing, M.~Mnich, and R.~Savani.
\newblock Reachability switching games.
\newblock {\em Logical Methods in Computer Science}, 17(2):10:1--10:29, 2021.

\bibitem{gartner2018CPS}
B.~G{\"a}rtner, T.~D. Hansen, P.~Hub{\'a}{\v{c}}ek, K.~Kr{\'a}l, H.~Mosaad, and
  V.~Sl{\'\i}vov{\'a}.
\newblock {ARRIVAL}: Next stop in {CLS}.
\newblock arXiv preprint
  \href{https://arxiv.org/abs/1802.07702v1}{arXiv:arXiv:1802.07702}, 2018.

\bibitem{gartner2021subexponential}
B.~G{\"a}rtner and H.~H.~P. Haslebacher, Sebastian.
\newblock A subexponential algorithm for {ARRIVAL}.
\newblock arXiv preprint
  \href{https://arxiv.org/abs/2102.06427v3}{arXiv:2102.06427}, 2021.

\bibitem{mayr1992fanout}
E.~W. Mayr and A.~Subramanian.
\newblock The complexity of circuit value and network stability.
\newblock {\em Journal of Computer and System Sciences}, 44(2):302--323, 1992.

\end{thebibliography}

\end{document}